\documentclass[
submission
]{dmtcs-episciences}

\usepackage[utf8]{inputenc}
\usepackage{subfigure}
\newtheorem{theo}{Theorem}[section]

\newtheorem{lemma}[theo]{Lemma}

\usepackage[round]{natbib}

\author{Stefan Bard\affiliationmark{1}
  \and Thomas Bellitto\affiliationmark{2}
  \and Christopher Duffy\affiliationmark{3}
  \and Gary MacGillivray\affiliationmark{1}\thanks{Research supported by NSERC.}
  \and Feiran Yang\affiliationmark{1}}
\title{Complexity of locally-injective homomorphisms to tournaments}
\affiliation{
  Department of Mathematics and Statistics, University of Victoria, CANADA\\
  Deptartment of Mathematics and Computer Science, University of Southern Denmark, DENMARK\\
  Department of Mathematics and Statistics, University of Saskatchewan, CANADA}
\received{2017-10-27}

\revised{2018-8-23}

\accepted{2018-9-3}

\begin{document}
\publicationdetails{20}{2018}{2}{4}{4021}
\maketitle
\begin{abstract}
For oriented graphs $G$ and $H$, a homomorphism $f: G \rightarrow H$ is locally-injective if, for every $v \in V(G)$, it is injective when restricted to some combination of the in-neighbourhood and out-neighbourhood of $v$. Two of the possible definitions of local-injectivity are examined. In each case it is shown that the associated homomorphism problem is NP-complete when $H$ is a reflexive tournament on three or more vertices with a loop at every vertex, and solvable in polynomial time when $H$ is a reflexive tournament on two or fewer vertices.
\end{abstract}

\section{Introduction}

Given two graphs $G=(V_G,E_G)$ and $H=(V_H,E_H)$, a \emph{homomorphism} from $G$ to $H$ is a function $f:V_G\to V_H$ such that for every $uv \in E_G,$ $f(u)f(v) \in E_H$. 
A homomorphism from $G$ to $H$ is referred to as an \emph{$H$-colouring} of $G$ and the vertices of $H$ are regarded as \emph{colours}. 
The graph $H$ is called the \emph{target} of the homomorphism.
These definitions extend to directed graphs by requiring that the mapping  must preserve the existence as well as the direction of each arc.

A \emph{locally-injective homomorphism} $f$ from $G$ to $H$ is a homomorphism from $G$ to $H$ such that for every $v \in V$ the restriction of $f$ to $N(v)$ (or possibly $N[v] = N(v) \cup \{v\}$) is injective.
The complexity of locally-injective homomorphisms for undirected graphs has been examined by a variety of authors and in a variety of contexts \cite{inj7, inj6, inj1, inj2, inj3, inj4, inj5, inj9}.
Locally-injective homomorphisms of graphs find application in a range of areas including  bio-informatics \cite{appli1, appli2, appli3} and coding theory \cite{inj5}. 

Here we consider locally-injective homomorphisms of \emph{oriented graphs}, that is, directed graphs in which any two vertices are joined by at most one arc. 
Given a vertex $v$, an arc from $v$ to $v$ is called a \emph{loop}.
A directed graph with a loop at every vertex is called \emph{reflexive}; a directed graph with no loops is called \emph{irreflexive}. For a vertex $v$ of a directed graph $G$, let $N^-(v)$, respectively $N^+(v)$, denote the set of vertices $u$ such that $uv$, respectively $vu$, is an arc of $G$. Note that if there is a loop at $v$, then $v \in N^-(v)$ and $v \in N^+(v)$.

To define locally-injective homomorphisms of oriented graphs, one must choose the neighbourhood(s) on which the homomorphism must be injective.  
Up to symmetry, there are four natural choices:
\begin{enumerate}
	\item $N^-(v)$.
	\item $N^+(v)$ and also $N^-(v)$.
	\item $N^+(v)\cup N^-(v)$.
	\item $N^+[v]\cup N^-[v] = N^+(v)\cup N^-(v) \cup \{v\}$.
\end{enumerate}

For irreflexive targets, $(2)$, $(3)$ and $(4)$ are equivalent. 
Under $(4)$, adjacent vertices must always be assigned different colours, and hence whether or not the target contains loops is irrelevant. 
Therefore, we may assume that targets are irreflexive when considering $(4)$. 
Then, a locally-injective homomorphism to an irreflexive target satisfying $(4)$ is equivalent to a locally-injective homomorphism to the same irreflexive target under either $(2)$ or $(3)$. 
As such, we need not consider $(4)$ and are left with three distinct cases.

Taking $(1)$ as our injectivity requirement defines  \emph{in-injective homomorphism};  taking $(2)$ defines  \emph{ios-injective homomorphism}; and taking $(3)$ defines \emph{iot-injective homomorphism}.
Here  ``ios" and ``iot" stand for ``in and out separately" and ``in and out together" respectively.

The problem of in-injective homomorphism is  examined by MacGillivray, Raspaud, and Swarts in \cite{MacRas, inj8}. 
They give a dichotomy theorem for the problem of in-injective homomorphism to reflexive oriented graphs, and one for the  problem of in-injective homomorphism to irreflexive tournaments.
The problem of in-injective homomorphism to irreflexive oriented graphs $H$ is shown to be NP-complete when the maximum in-degree of $H$, $\Delta^-(H)$, is at least $3$, and solvable in polynomial time when $\Delta^-(H)=1$. 
For the case $\Delta^-(H)=2$ they show that an instance of directed graph homomorphism polynomially transforms to an instance of in-injective homomorphism to a target with maximum in-degree $2$. 
As such the restriction of  in-injective homomorphism to targets $H$ so that $\Delta^-(H)=2$ constitutes a rich class of problems. 

The remaining problems, ios-injective homomorphism and iot-injective homomorphism, are considered by Campbell, Clarke and MacGillivray \cite{Russell, Rus1, Rus2}. In this paper we extend the results of Campbell, Clarke and MacGillivray to provide dichotomy theorems for the restriction of the problems of iot-injective homomorphism and ios-injective homomorphism to reflexive tournaments.

Preliminary results are surveyed in Section 2. In Section 3, we show that ios-injective homomorphism is NP-complete for reflexive tournaments on $4$ or more vertices. In Section 4, we show that iot-injective homomorphism is also NP-complete for reflexive tournaments on $4$ or more vertices. We close with a brief discussion of injective homomorphisms to irreflexive tournaments.

\section{Known Results}

For a fixed undirected graph $H$, the problem of determining whether an undirected graph $G$ admits a homomorphism to $H$ (i.e., \emph{the $H$-colouring problem}) admits a well-known dichotomy theorem.

\begin{theo}[\cite{HellNesetril}]
	Let $H$ be an undirected graph. If $H$ is irreflexive and non-bipartite, then $H$-colouring is NP-complete. If $H$ has a loop, or is bipartite, then $H$-colouring is solvable in polynomial time.
\end{theo} 

A dichotomy theorem for the complexity of $H$-colouring of directed graphs is given by Bulatov \cite{Bulat} and Zhuk \cite{Zhuks}.

For fixed small reflexive tournaments $T$,  Campbell, Clarke and MacGillivray give the following result for the complexity of  ios-injective $T$-colouring and iot-injective $T$-colouring, where ios-injective $T$-colouring and iot-injective $T$-colouring are defined analogously to $H$-colouring.

\begin{theo}[\cite{Russell, Rus1, Rus2}]
	\label{rus1}
	
	If $T$ is a reflexive tournament on $2$ or fewer vertices, then ios-injective $T$-colouring and iot-injective $T$-colouring are solvable in polynomial time. If $T$ is a reflexive tournament on $3$ vertices, then ios-injective $T$-colouring and iot-injective $T$-colouring are NP-complete.
	
\end{theo}

\begin{figure}[!ht]
\begin{center}
		\includegraphics{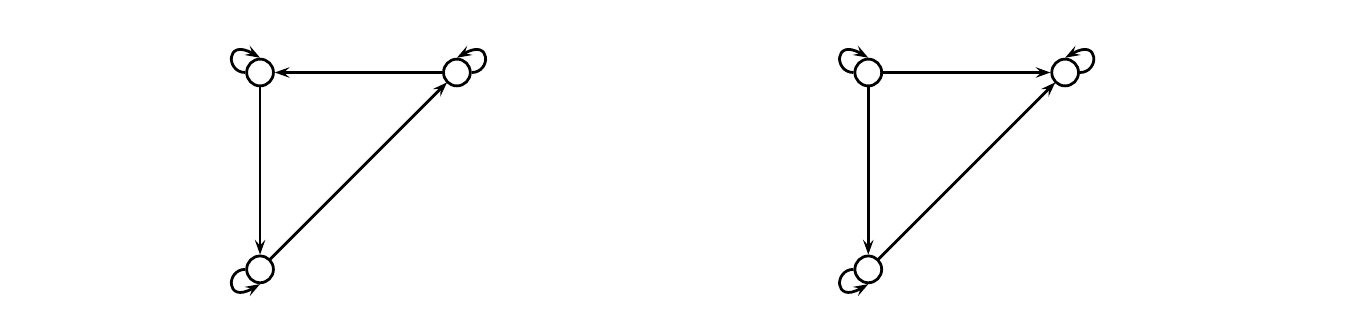}
	\caption{The two reflexive tournaments on three vertices: the reflexive three-cycle $C_3$ and the reflexive transitive tournament on three vertices $T\!T_3$.}
	\label{Sec2Graps}
\end{center}
\end{figure}

\section{Ios-injective homomorphisms}\label{sec:ios}

In this section we prove a dichotomy theorem for ios-injective $T$-colouring, where $T$ is a reflexive tournament. 
We show that both ios-injective $T_4$-colouring and ios-injective $T_5$-colouring are NP-complete (see Figures \ref{T4} and \ref{T5}).
We then show that any instance of ios-injective $T$-colouring, where $T$ is a reflexive tournament on at least $4$ vertices, polynomially transforms to an instance of
ios-injective $T^\prime$-colourings, where $T^\prime$ is $T_4$, $T_5$, $C_3$ or $T\!T_3$ (see Figures \ref{Sec2Graps}, \ref{T4} and \ref{T5}).
The dichotomy theorem follows from combining these results with the result in Theorem~\ref{rus1}.

We begin with a study of ios-injective $T_4$-colouring.
To show ios-injective $T_4$-colouring is NP-complete we provide a transformation from 3-edge-colouring subcubic graphs.
We  construct an oriented graph $H$ from a  graph $G$ so that $G$ has a 3-edge-colouring if and only if $H$ admits an ios-injective homomorphism to $T_4$.
The key ingredients in this construction are a pair of oriented graphs, $H_x$ and $H_{e}$, given in Figures \ref{VGT4IOS} and \ref{EGT4IOS}, respectively. Figures \ref{edgegadgios} and \ref{fig:vxcolour} give ios-injective $T_4$-colourings of $H_x$ and $H_e$, respectively.

%
%

\begin{figure}[!ht]
	\begin{center}
	\includegraphics{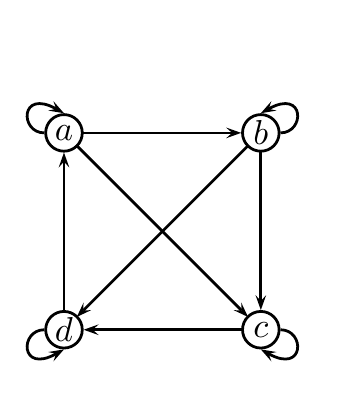}
	\caption{$T_4$ -- the only strongly connected reflexive tournament on four vertices. }
	\label{T4}
	\end{center}
\end{figure}

\begin{figure}[!ht]
	\begin{center}
	\includegraphics{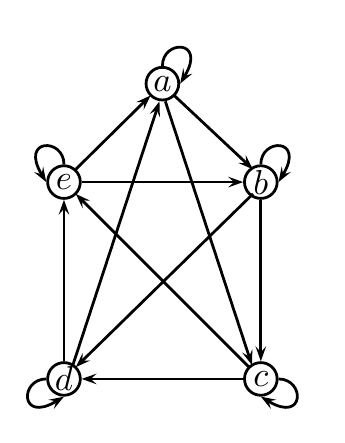}
	\caption{$T_5$ -- the only reflexive tournament on five vertices where all the vertices have in-degree and out-degree three.}
	\label{T5}
	\end{center}
\end{figure}

\begin{figure}
	\begin{center}
		\includegraphics[scale=1]{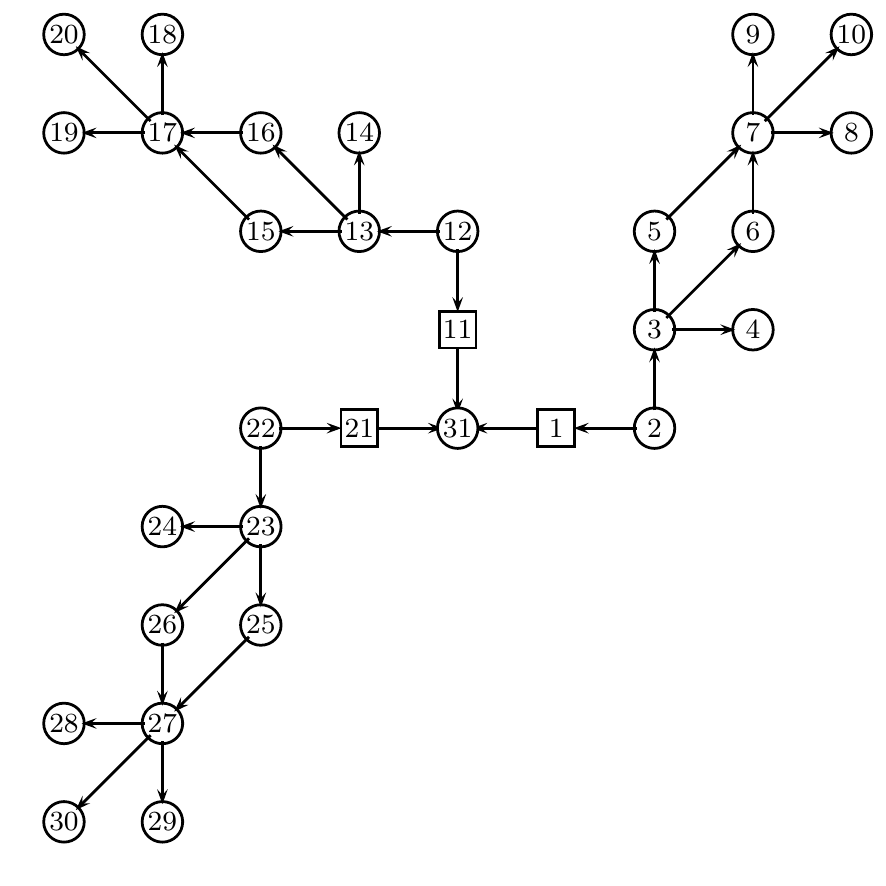}
	\caption{$H_x$.}
	\label{VGT4IOS}
	\end{center}
\end{figure}

\begin{figure}
	\begin{center}
		\includegraphics{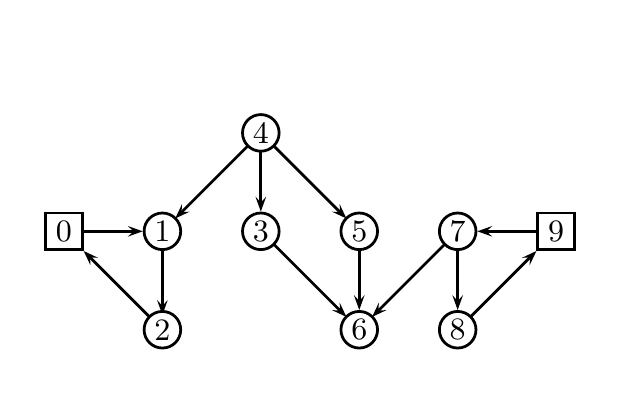}
	\caption{$H_{e}$.}
	\label{EGT4IOS}
	\end{center}
\end{figure}

\begin{lemma} \label{lem:vxgadget1}
	In any ios-injective $T_4$-colouring of $H_x$
	\begin{enumerate}
		\item the vertices $3,13$ and $23$ are coloured $a$; and
		\item vertex $31$ is coloured $d$.
	\end{enumerate} 
\end{lemma}

\begin{proof}
	\noindent\emph{(1)} By symmetry, it suffices to show the claim for vertex 3. Let us first note that the vertices 3 and 7 have out-degree 3 and can therefore only be coloured $a$ or $b$, as these are the only vertices of out-degree  $3$ in $T_4$. 
	If vertex 7 is coloured $a$, then its two in-neighbours, vertices 5 and 6, are coloured $d$ and $a$. 
	However, this  is impossible as no vertex of out-degree three in $T_4$ has both $d$ and $a$ as out-neighbours. 
	Hence, vertex 7 is coloured $b$. 
	If vertex 3 is coloured $b$, then vertices 5 and 6 would be both in- and out-neighbours of vertices coloured $b$.
	Thus, each of vertices $5$ and $6$ are coloured $b$.
	This is a violation of the injectivity requirement.
	Therefore  vertex 3 (and by symmetry, the vertices 13 and 23) must be coloured $a$. \\
	
	\noindent\emph{(2)} Notice that the square vertices in the $H_x$ (vertices 1, 11 and 21) cannot be coloured $a$;  they each have an in-neighbour that already has an out-neighbour coloured $a$. 
	These square vertices have a common out-neighbour and so must receive distinct colours by the injectivity requirement.
	As none is coloured $a$, these three vertices are coloured $b$, $c$ and $d$, in some order.
	The only vertex that is an out-neighbour of $b$, $c$ and $d$ in $T_4$ is $d$. 
	And so, the common out-neighbour of vertices 1, 11 and 21 (i.e., vertex 31) has colour $d$.
\end{proof}

\begin{lemma} \label{lem:egadget1}
	Let $H_e^\prime$ be an oriented graph formed from a copy of $H_{e}$ and two copies of $H_x$ by identifying vertex $0$ in $H_e$ with any square vertex in one copy of $H_x$ and identifying vertex $9$ in $H_e$ with any square vertex in the second copy of $H_x$. In any ios-injective $T_4$-colouring of $H_e^\prime$, the vertices $0$ and $9$ in the subgraph induced by $H_e$ have the same colour.
\end{lemma}

\begin{proof}
	Let $H_e^\prime$ be constructed as described.
	Consider an ios-injective $T_4$-colouring of $H_e^\prime$.
	We examine the colours of the vertices in the subgraph induced by the copy of $H_e$.
	By Lemma \ref{lem:vxgadget1} and the construction of  $H_e^\prime$, vertices $0$ and $9$ each have an in-neighbour that has an out-neightbour coloured $a$.
	By the injectivity requirement neither vertex $0$ or $9$ is coloured $a$. 
	We proceed in cases to show that vertices $0$ and $9$ receive the same colour.
	
	\emph{Case I: Vertex 0 is coloured $b$.}
	It cannot be that vertex 1 is coloured $d$, as vertex 0  has an out-neighbour coloured $d$ -- vertex 31 in a copy of $H_x$.
	It cannot be that vertex 1 is coloured $c$ as no 3-cycle of $T_4$ contains both a vertex coloured $b$ and a vertex coloured $c$. 
	Thus,  vertex 1 must be coloured $b$. 
	The vertex 2 is both an in-neighbour and an out-neighbour of vertices coloured $b$ and is therefore coloured $b$. 
	The vertex 4 is an in-neighbour of vertex 1, and so cannot be coloured $b$ as vertex 1 already has an in-neighbour coloured $b$. 
	The vertex 4 must thus be coloured $a$. 
	By injectivity, the out-neighbours of vertex 4 must receive distinct colours that are out-neighbours of $a$ in $T_4$.
	Therefore vertices 3 and 5 are coloured $a$ and $c$ in some order, as vertex $1$ is coloured $b$.
	The only common out-neighbour of $a$ and $c$ in $T_4$ is $c$. 
	As such, vertex 6 must be coloured $c$. 
	By injectivity, each of the in-neighbours of vertex $6$ must receive distinct colours that are in-neighbours of $c$ in $T_4$.
	And so vertex 7 must be coloured $b$. 
	As vertex 9 cannot be coloured $a$ and it has an out-neighbour coloured $b$, namely vertex 7, we have that vertex 9 must be coloured $b$.
	Thus vertices 0 and 9 have the same colour.
	
	\emph{Case II: Vertex 0 is coloured $c$.}
	The out-neighbours of $c$ in $T_4$ are $c$ and $d$.
	It cannot be that vertex 1 is coloured $d$, as vertex 0  has an out-neighbour coloured $d$ -- vertex 31 in a copy of $H_x$.
	And so vertex $1$ is coloured $c$.
	
	The vertex $4$ has an out-neighbour coloured $c$, and so must be coloured  $a$ or $b$ or $c$.
	Since vertex 0 is coloured $c$, vertex 4 cannot be coloured $c$ without violating injectivity. We claim vertex 4 is coloured $b$.

	If vertex 4 is coloured $a$, then by injectivity, vertices 3 and 5 are coloured $a$ and $b$, in some order.
	The only out-neighbour of $a$ and $b$ in $T_4$ that has in-degree 3 is $c$. 
	As such vertex 6 is coloured $c$. 
	The vertex $c$ in $T_4$ has three in-neighbours -- $a$, $b$, and $c$.
	As vertex 6 has in-neighbours coloured $a$ and $b$ (namely, vertices 3 and 5), then by injectivity the third in-neighbour of vertex 6 (namely, vertex 7) is  coloured $c$.
	In $T_4$, $c$ has two out-neighbours: $c$ and $d$.
	Since vertex 7 is coloured $c$ and already has an out-neighbour coloured $c$,  vertex 8 must be coloured $d$. 
	The vertex $9$ has an in-neighbour coloured $d$.
	Only vertices $a$ and $d$ in $T_4$ have $d$ as an in-neighbour.
	Therefore vertex $9$ is coloured with $a$ or $d$.
	However, we have shown previously that vertex $9$ cannot be coloured $a$.
	This implies, that vertex $9$ is coloured $d$.
	However, vertex $9$ has an out-neighbour coloured $c$.
	Since $c$ is not an out-neighbour of $d$ in $T_4$, we arrive at a contradiction.
	Thus vertex 4 is not coloured $a$.
	Therefore vertex $4$ is coloured $b$.
	
	Since vertex $4$ is coloured $b$,  vertices 3 and 5 are coloured $b$ and $d$, in some order. 
	The only common out-neighbour of $b$ and $d$ in $T_4$ is $d$.
	Therefore  vertex 6 is  coloured $d$. 
	Hence, by injectivity, vertex 7  is coloured $c$.
	Since vertex $7$ has an out-neighbour coloured $d$,  vertex 8, another out-neighbour of vertex $7$ must be coloured $c$. 
	Since 9 has both an in-neighbour and an out-neighbour coloured $c$, vertex 9 must be coloured $c$.
	Thus vertices 0 and 9 have the same colour.
	
	\emph{Case III: Vertex 0 is coloured $d$.}
	It cannot be that vertex 1 is coloured $d$, as vertex 0  has an out-neighbour coloured $d$ -- vertex 31 in a copy of $H_x$.
	Vertex $d$ has two out-neighbours in $T_4$: $a$ and $d$.
	Therefore vertex $1$ is coloured $a$.
	
	The vertex 4 has out-degree 3  and an out-neighbour coloured $a$.
	Vertex $a$ is the only vertex in $T_4$ to have out-degree 3 and have $a$ as an out-neighbour.
	Therefore vertex $4$ is coloured $a$.
	By injectivity the vertices 3 and 5, the remaining out-neighbours of vertex $4$,  are coloured $b$ and $c$. 
	The vertex 7 cannot be coloured $d$ since 9 already has an out-neighbour coloured $d$ -- vertex 37 in a copy of $H_x$.
	Moreover, vertex 7 is an in-neighbour of 6, which already has in-neighbours coloured $c$ and $b$. 
	Hence, vertex 7 must be coloured $a$. 
	In $T_4$ the only in-neighbours of $a$ are $a$ and $d$.
	Thus vertex $9$ is coloured $a$ or $d$.
	Since vertex 9 has an out-neighbour coloured $d$, it cannot be coloured $a$.
	Therefore vertex $9$ is coloured $d$, as required.
\end{proof}


\begin{theo}\label{thm:main4}
	The problem of ios-injective $T_4$-colouring is NP-complete.
\end{theo}

\begin{proof}
	The transformation is from  3-edge-colouring subcubic graphs \cite{Holyer}.

	Let $G$ be a graph with maximum degree at most $3$ and let $\tilde{G}$ be an arbitrary orientation of $G$.
	We create an oriented graph $H$ from $\tilde{G}$ as follows.
	For every $x \in V(G)$ we add $H_x$, a copy of the oriented graph given in Figure  \ref{VGT4IOS}, to $H$.
	For every arc $e \in E(\tilde{G})$ we add $H_{e}$, a copy of the oriented graph given in  Figure \ref{EGT4IOS}, to $H$.
	To complete the construction of $H$, for each arc $e = uv \in E(\tilde{G})$ we identify vertex $0$ in $H_{e}$ with one of the three square vertices (i.e., vertices $1$, $11$, or $21$) in $H_u$ and identify vertex $9$ in $H_{e}$ with one of the three square vertices in $H_v$.
	We identify these vertices in such a way that each square vertex in a copy of $H_x$ is identified with at most one square vertex from a copy of $H_e$.
	We note that this is always possible as vertices in $G$ have degree at most three.

	
	We claim  $G$ has a  3-edge-colouring if and only if $H$ has an ios-injective $T_4$-colouring.
	Suppose an ios-injective $T_4$-colouring of $H$ is given. 
	This ios-injective $T_4$-colouring induces a 3-edge-colouring of $G$: the colour of an edge in $e \in E(G)$ is given by colour of vertices $0$ and $9$ in corresponding copy of $H_{e}$ contained in $H$. 
	By Lemma \ref{lem:egadget1}, this colour is well-defined.
	By Lemma \ref{lem:vxgadget1}, vertex 31 in each copy of $H_e$ is coloured $d$. Therefore, each of the edges incident with any vertex of $G$ receive different colours and no more than $3$ colours, namely $b$, $c$, and $d$, are used on the edges of $G$.
	
	\begin{figure}
		\begin{center}
			\includegraphics[width = .3\linewidth]{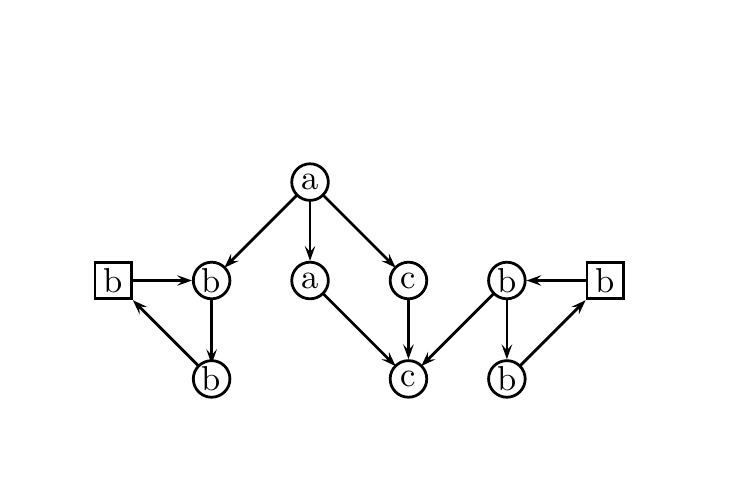}
		\includegraphics[width = .3\linewidth]{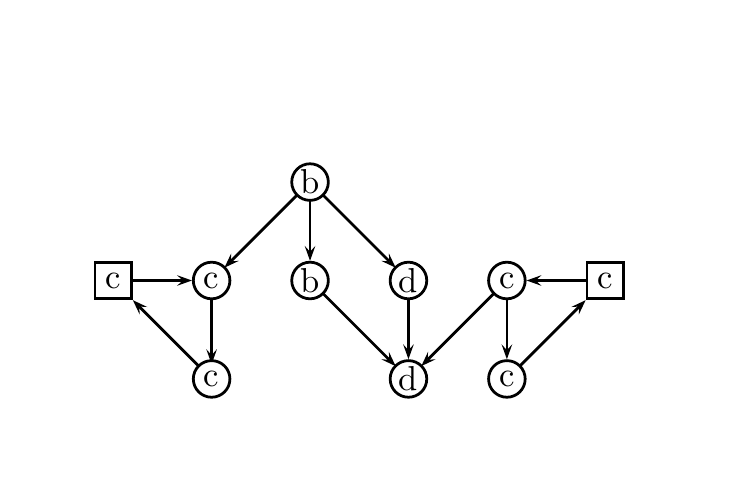}
		\includegraphics[width = .3\linewidth]{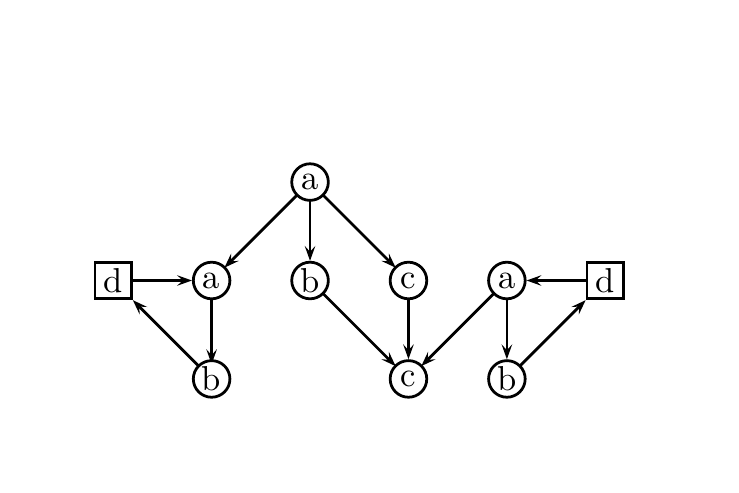}
		\caption{Colouring the copies of $H_{e}$ in  Theorem \ref{thm:main4}.}
		\label{edgegadgios}
		\end{center}
	\end{figure}
	
	
	
	
	Suppose a 3-edge-colouring of $G$, $f: E(G) \to \{b,c,d\}$ is given.
	For each $e \in E(G)$ we colour $H_{e}$ using one of the ios-injective $T_4$-colourings given in Figure~\ref{edgegadgios}.
	We choose the colouring of each copy of $H_{e}$ so that vertices $0$ and $9$ in that copy are assigned the colour $f(e)$.
	To complete the proof, we show that such a colouring can be extended to all copies of $H_x$ contained in $H$.
	
	Recall for each copy of $H_x$, the vertices 1, 11 and 21 can be respectively identified with either vertex 0 or vertex 9 in some copy of $H_{e}$.
	Since $f$ is a 3-edge-colouring of $G$, for each $x \in V(G)$, each of the vertices 1, 11 and 21 in $H_x$ are coloured with distinct colours from the set $\{b,c,d\}$ when we colour each copy of $H_{e}$ using Figure~\ref{edgegadgios}.
	
	By symmetry of $H_x$, we may assume without loss of generality that vertices $1,11$ and $21$ are respectively coloured $b$, $c$ and $d$ in each copy of $H_x$.
	The ios-injective $T_4$-colouring given in Figure \ref{fig:vxcolour} extends a pre-colouring of the vertices 1, 11 and 21 with colours $b$, $c$, and $d$, respectively, to an ios-injective $T_4$-colouring of $H_x$. 
	Therefore $G$ has a 3-edge-colouring if and only if $H$ admits an ios-injective $T_4$-colouring
	
	Since the construction of $H$ can be carried out in polynomial time,  ios-injective $T_4$-colouring is NP-complete
\end{proof}

\begin{figure}[!ht]
	\begin{center}
		\includegraphics[scale = 1]{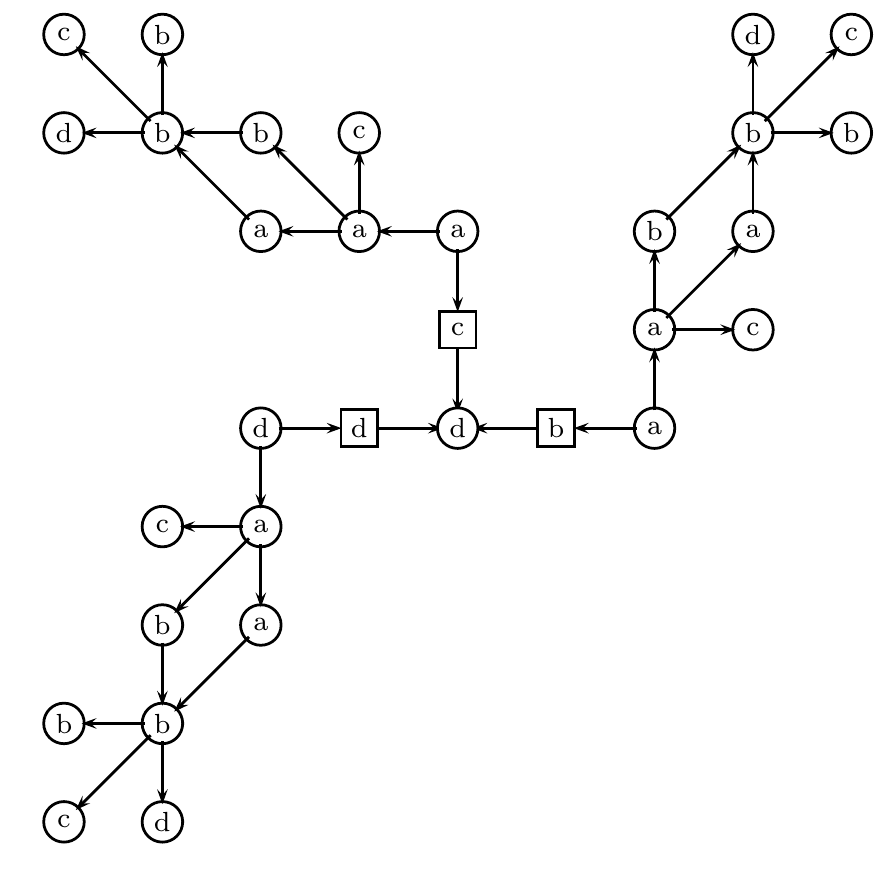}
	\caption{Colouring the copies of $H_x$ in  Theorem \ref{thm:main4}.}
	\label{fig:vxcolour}
	\end{center}
\end{figure}

We give a similar argument for $T_5$ (see Figure~\ref{T5}). 
The transformation is from ios-injective $C_3$-colouring, which is NP-complete by Theorem~\ref{rus1}.  
We construct an oriented graph $J$ from a  graph $G$ so that $G$ admits an ios-injective homomorphism to $C_3$ if and only if 
$J$ admits an ios-injective homomorphism to $T_5$.
The key ingredient in this construction is the oriented graph $J_v$, given in Figure \ref{VGT5IOS}. 

For each $n>0$ we construct an oriented graph $J_n$ from $n$ copies of $J_v$, say $J_{v_0}, J_{v_1}, \dots, \\ J_{v_{n-1}}$, by letting vertices 17, 18 and 19 of $J_{v_i}$ be in-neighbours of vertex $0$ in $J_{v_{i+1}\pmod n}$ for all $0 \leq i \leq n-1$.

\begin{figure}
\begin{center}
		\includegraphics{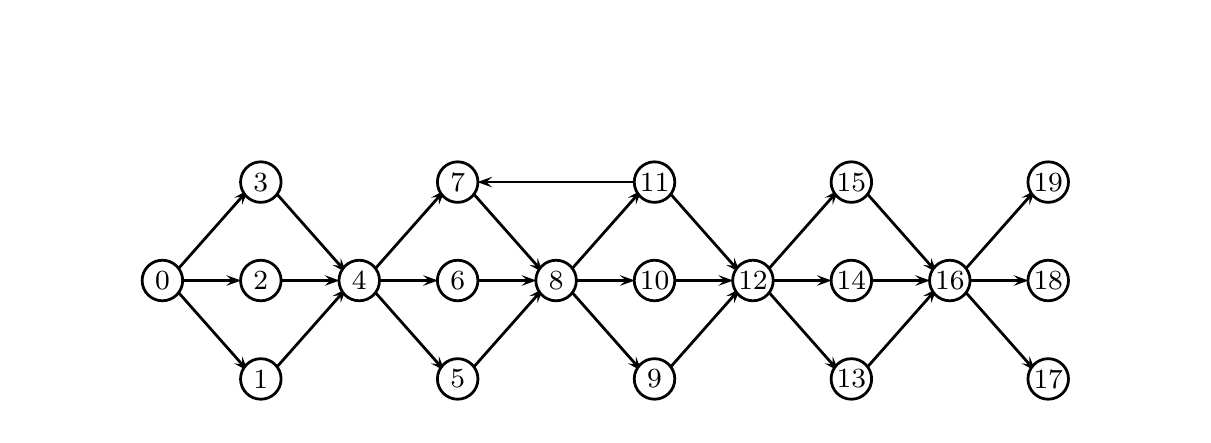}
	\caption{$J_v$.}
	\label{VGT5IOS}
\end{center}
\end{figure}

\begin{lemma} \label{lem:loopLemma}
	For any positive integer $n$, in an oriented ios-injective $T_5$-colouring of $J_n$ each of the vertices labelled $0$ (respectively, $4$, $8$, $12$ and $16$) receives the same colour.
\end{lemma}

\begin{proof}
	Since $T_5$ is vertex-transitive, assume without loss of generality that vertex $0$ in 	$J_{v_0}$ receives colour $a$. 
	If 0 is coloured $a$, then the vertices 1, 2 and 3 must be coloured $a,b$ and $c$ in some order, these vertices are the only out-neighbours of $a$ in $T_5$. Since vertex $c$ is the only common out-neighbour of vertices $a$, $b$ and $c$ in $T_5$ we have that vertex $4$ is coloured $c$. 
	Since the automorphism of $T_5$ that maps $a$ to $c$ also maps $c$ to $e$, we conclude by a similar argument that vertex $8$ is coloured $e$. 
	Similarly we conclude that vertex $12$ is coloured $b$ and vertex $16$ is coloured $d$. 
	
	Since vertex $16$ is coloured $d$ in $J_{v_0}$, vertices $17,18,19$ are coloured, in some order, $a,e,d$, as these are the only out-neighbours of $d$ in $T_5$.
	The only common out-neighbour of $a$, $e$ and $d$ in $T_5$ is $a$. 
	Therefore vertex $0$ in $J_{v_1}$ is coloured $a$. 
	Repeating this argument, we conclude that vertices 4, 8, 12 and 16 in $J_{v_2}$ receive colours $c$, $e$, $b$ and $d$, respectively. Continuing in this fashion gives that in an oriented ios-injective $T_5$-colouring of $J_n$ each of the vertices labeled $0$ (respectively, $4$, $8$, $12$ and $16$) receive the same colour.	
\end{proof}



\vspace{0.35cm}
\begin{theo} \label{thm:main5}
	The problem of ios-injective $T_5$-colouring is NP-complete.
\end{theo}

\begin{proof}
	The transformation is from ios-injective $C_3$-colouring (See Theorem~\ref{rus1}).
	
	Let $G$ be a graph with vertex set $\{v_0,v_1, \dots, v_{|V(G)|-1}\}$.
	Let $\nu_G = |V(G)|$.
	We construct $J$ from $G$ by first adding a copy of $J_{\nu_G}$ to $G$  and then, for each $1 \leq i \leq \nu_G$, adding an arc from vertex $11$ in $J_{v_i}$ to $v_i$.
	
	
	We show that $J$ has an ios-injective $T_5$-colouring if and only if $G$ has an ios-injective $C_3$-colouring.
	
	Consider an ios-injective $T_5$-colouring of $J$. Since $T_5$ is vertex-transitive we can assume without loss of generality that vertex 8 in each copy of $J_v$ is coloured $a$. Therefore in each $J_{v_i}$, vertices $9,10$ and $11$ are coloured, in some order, with colours $a,b,c$, and vertex $12$ is coloured $c$.
	
	We claim that $v_i$ is coloured with $b$, $d$ or $e$ for all $0 \leq i \leq \nu_G-1$.
	If $v_i$ has colour $a$, then vertex $11$ in $J_{v_i}$ has both an in-neighbour and an out-neighbour coloured $a$ and is therefore coloured $a$. Thus, vertex $7$ in $J_{v_i}$ also has both an in-neighbour and an out-neighbour coloured $a$ and must be coloured $a$. Since vertex 11 already has an out-neighbour coloured $a$, this contradicts the injectivity requirement.
	If $v_i$ has colour $c$, then vertex $11$ in $J_{v_i}$ has two out-neighbours coloured $c$, a violation of the injectivity requirement.
	Therefore $v_i$ is coloured with one of $b$, $d$ or $e$ for each $0 \leq i \leq \nu_G-1$. 
	Since vertices $b$, $d$ and $e$ of $T_5$ induce a copy of $C_3$ in $T_5$, restricting   an ios-injective $T_5$-colouring of $J$ to the vertices of $G$ yields an ios-injective $C_3$-colouring of $G$.
	
	Let $\beta$ be an ios-injective $C_3$-colouring of $G$ using colours $b$, $d$ and $e$. We extend such a colouring to be an ios-injective $T_5$-colouring of $J$ by assigning the vertices of each $J_{v_i}$ colours based upon $\beta(v_i)$ as shown in Figure~\ref{fig:t5gadg}. This colouring satisfies the injectivity requirement, as each vertex $v_i$ has only neighbours coloured $b,d$ and $e$ in $G$, and its additional neighbour in $J_{v_i}$, vertex $11$, has colour $a$ or $c$.

	\begin{figure}[!ht]
	\begin{center}
			\includegraphics[width = .6\linewidth]{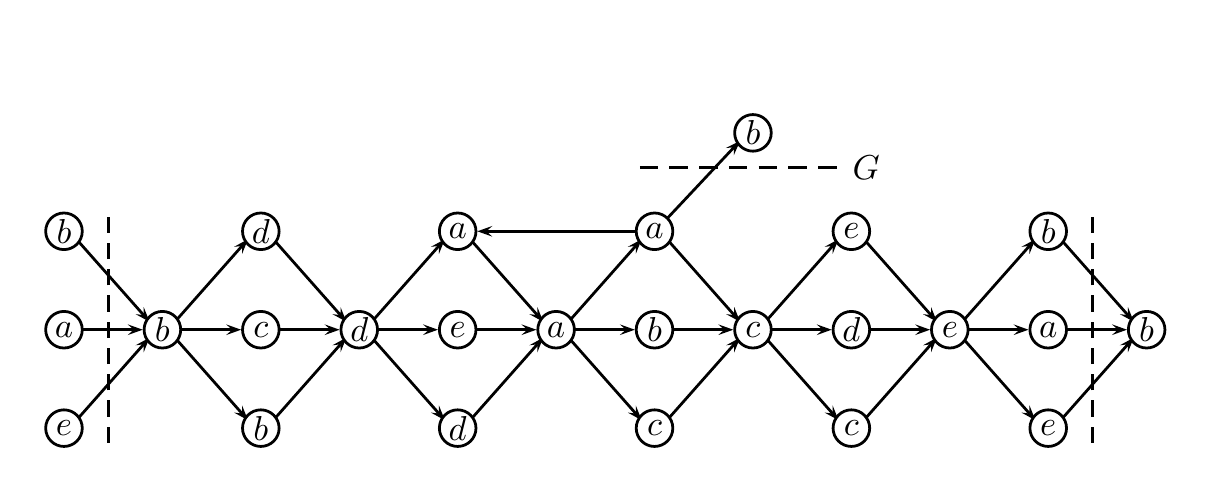}
		\includegraphics[width = .6\linewidth]{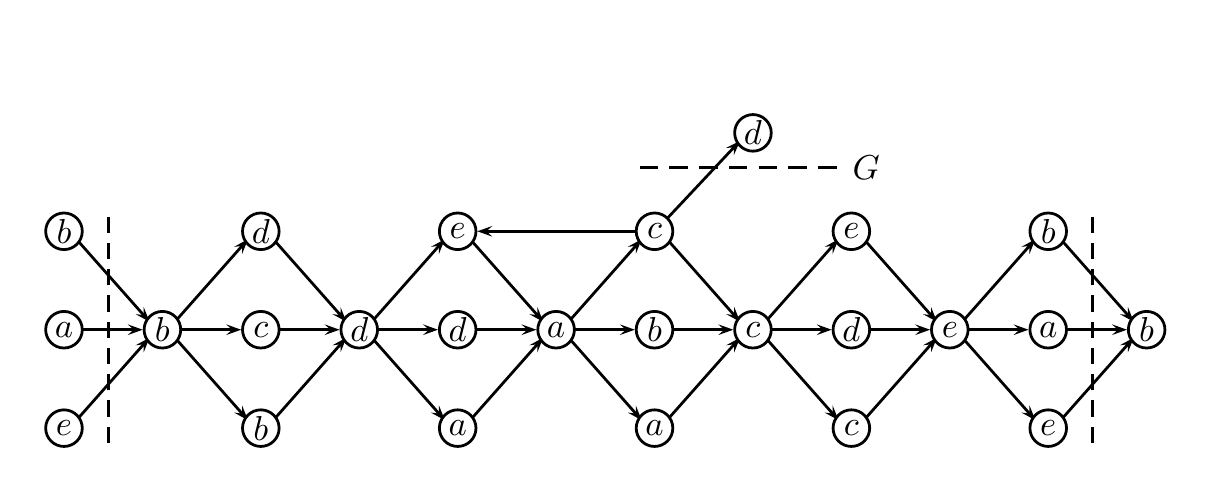}
		\includegraphics[width = .6\linewidth]{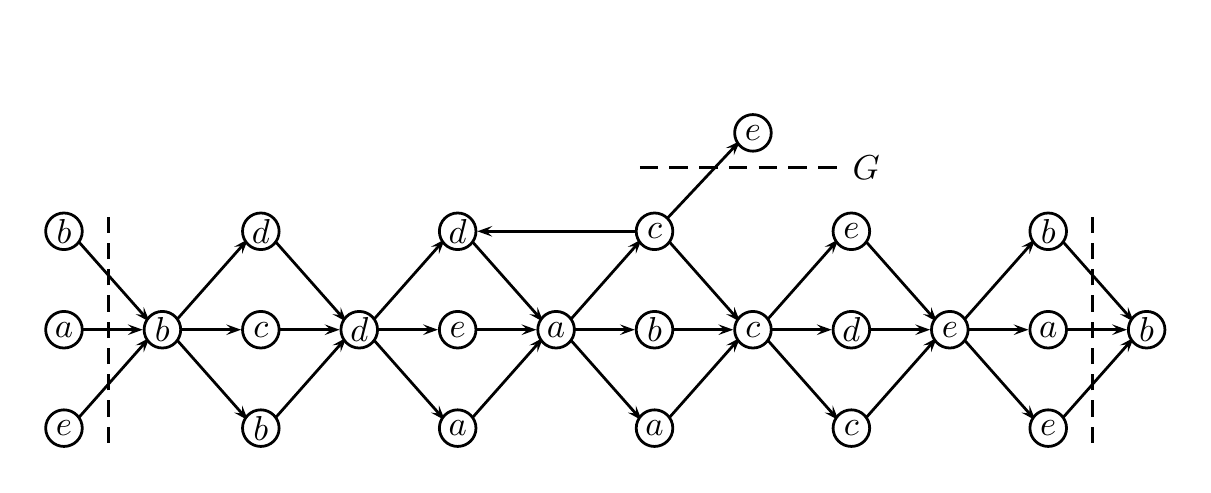}
		\caption{Colouring the vertices of $J_{v_i}$ using the colour of $v_i$.}
		\label{fig:t5gadg}
	\end{center}
	\end{figure}
	
	
	
	Therefore $J$ has an ios-injective $T_5$-colouring if and only if $G$ has an ios-injective $C_3$-colouring. 
	Since $J$ can be constructed in polynomial time, ios-injective $T_5$-colouring is NP-complete.
\end{proof}

We now present a reduction to instances of ios-injective $T$-colouring  for when $T$ has a vertex $v$ of out-degree  at least four. This reduction allows us to polynomially transform an instance of ios-injective $T$-colouring to an instance of ios-injective $T^\prime$-colouring, where $T^\prime$ is $T_4$, $T_5$, $C_3$ or $T\!T_3$.

\begin{lemma}
	\label{indios}
	If $T$ is a reflexive tournament on $n$ vertices with a vertex $v$ of out-degree at least four, then ios-injective homomorphism to $T'$ polynomially transforms to ios-injective homomorphism to $T$, where $T'$ is the tournament induced by the strict out-neighbourhood of $v$. 
\end{lemma}

\begin{proof}
	Let $T$ be a reflexive tournament on $n$ vertices with a vertex $v$ of out-degree at least four.
	Let $G$ be an oriented graph with vertex set $\{w_0,w_1,\dots,w_{|V(G)|-1}\}$.
	Let $\nu_G = |V(G)|$.
	We construct $H$ from $G$ by adding to $G$
	\begin{itemize}
		\item vertices $x_0,x_1, \dots,x_{\nu_G-1}$.
		\item an arc from $x_i$ to $w_i$ for all $0 \leq i \leq \nu_G-1$.
		\item  $\nu_G$ irreflexive copies of $T$, labeled $T_{i}$, for $0 \leq i \leq \nu_G-1$.
	\end{itemize}
	Let $v_i\in T_{i}$ be the vertex corresponding to $v \in V(T)$.
	We complete our construction by adding the arcs $v_ix_i$ and $x_i v_{i+1\pmod{\nu_G}}$ for all $i$.
	See Figure \ref{ioslem}.
	
	\begin{figure}
		\begin{center}
			
		\includegraphics{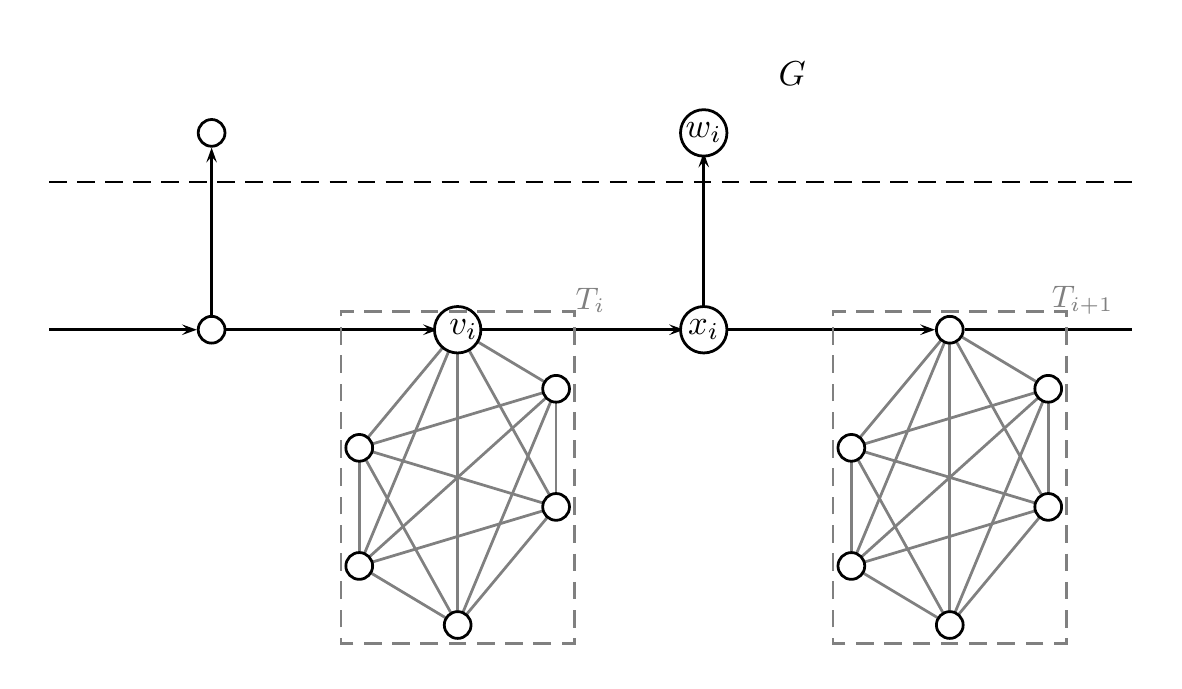}
		\caption{The construction of $H$ in Lemma \ref{indios}.}
		\label{ioslem}
		\end{center}
	\end{figure}
	
	\emph{Claim 1: In an ios-injective $T$-colouring of $H$ no two vertices of $T_{i}$ have the same colour}.
	Since $T$ has a vertex of out-degree at least four we observe that $T$ has at least $4$ vertices.
	Let $\phi$ be an ios-injective $T$-colouring of $H$.
	Suppose there exist $x,y \in T_{i}$ so that $\phi(x) = \phi(y) = c$, and let $z$ be a third vertex of $T_i$. 
	By injectivity, $x$ and $y$ cannot be either both in-neighbours or both out-neighbours of $z$, and therefore $z$ has an in-neighbour and an out-neighbour coloured $c$. This is only possible if $\phi(z) = c$.
	Let $w$ be a fourth vertex of $T_i$. Since $x$, $y$ and $z$ are all neighbours of $w$, $w$ has either two in-neighbours or two out-neighbours of the same colour, which is impossible in an ios-injective $T$-colouring. This proves the claim.
	
	
	
	

	\emph{Claim 2: In an ios-injective $T$-colouring of $H$, every vertex in the set $$\{x_0,x_1, \dots, x_{\nu_G-1}\} \cup \{  v_0,v_1, \dots, v_{\nu_G-1}  \}$$ receives the same colour}.
	By the previous claim all the colours of $T$ are used exactly once in each $T_{i}$.
	Therefore the only possible colour for an out-neighbour or an in-neighbour of $v_i$ outside of $T_i$ is $\phi(v_i)$. Therefore for each $i$, we have $\phi(v_i) = \phi(x_i) = \phi(x_{i+1 \pmod{\nu_G}})$.
	Therefore every vertex in the set $\{x_0,x_1, \dots, x_{\nu_G-1}\} \cup \{  v_0,v_1, \dots, v_{\nu_G-1}  \}$ receives the same colour.
	Since each vertex in $T_{i}$ maps to a unique vertex in $T$, if $\phi(v_i) \neq v$ then there is an automorphism of $T$ that maps $\phi(v_i)$ to $v$. As such, we may assume without loss of generality that $\phi(v_i) = v$ for all $0 \leq i \leq \nu_G-1$.

	Let $T^\prime$ be the reflexive sub-tournament of $T$ induced by the strict out-neighbourhood of $v$. Note that $T'$ is a reflexive tournament on  at least 3 vertices and  $T^\prime$ is a proper subgraph of $T$.
	
	We show that $H$ has an ios-injective $T$-colouring if and only if $G$ has an ios-injective $T'$-colouring.
	
	Let $\phi$ be an ios-injective $T$-colouring of $H$.
	By our previous claim, each $x_i$ has an in-neighbour and an out-neighbour with colour $v$, namely $v_i \in V(T_{i})$ and $v_{i-1} \in V(T_{{i-1}})$.
	Therefore $\phi(w_i)$ is an out-neighbour of $v$ in $T$.
	That is, $\phi(w_i) \in V(T^\prime)$.
	Therefore the restriction of $\phi$ to the vertices of $G$ yields an ios-injective $T^\prime$-colouring of $G$.

	Let $\beta$ be an ios-injective $T'$-colouring of $G$.
	For all $0 \leq i \leq \nu_G-1$ and all $u \in V(T)$. Let $u_i \in V(T_i)$ be the vertex corresponding to $u \in V(T)$.
	
	We  extend $\beta$ to be  an ios-injective $T$-colouring of $H$ as follows:
	\begin{itemize}
		\item $\beta(x_i) = \beta(v_i) =v$ for all $0 \leq i \leq \nu_G-1$; and
		\item for all $u_i \in T_{i}$, let $\beta(u_i) = u$.
	\end{itemize}
	
	Hence, ios-injective $T'$-colouring of $G$ can be polynomially transformed to ios-injective $T$-colouring of $H$. 
\end{proof}

If $T$ is a reflexive tournament with a vertex of in-degree at least 4, a similar argument holds. We modify the construction by reversing the arc between $x_i$ and $w_i$ in the construction of $H$.

\begin{lemma}
	\label{indios2}
	If $T$ is a reflexive tournament on $n$ vertices with a vertex $v$ of in-degree at least four, then ios-injective homomorphism to $T'$ polynomially transforms to ios-injective homomorphism to $T$, where $T'$ is the tournament induced by the strict in-neighbourhood of $v$. 
\end{lemma}

Our results compile to give a dichotomy theorem.

\begin{theo}
	\label{iosdichot}
	Let $T$ be a reflexive tournament. If $T$ has at least $3$ vertices, then the problem of deciding whether a given oriented graph $G$ has an ios-injective homomorphism to $T$ is NP-complete. If $T$ has $1$ or $2$ vertices, then the problem is solvable in polynomial time.
\end{theo}

\begin{proof}
	If $T$ is a reflexive tournament on no more than three vertices, the result follows by Theorem \ref{rus1}. Suppose then that $T$ has four or more vertices. 
	If $T = T_4$, or if $T = T_5$, then the result follows from Theorem \ref{thm:main4} or Theorem \ref{thm:main5}. Up to isomorphism, there are $16$ distinct reflexive tournaments on $4$ or $5$ vertices. By inspection, tournaments $T_4$ and $T_5$ respectively are the only reflexive tournaments on $4$ and $5$ vertices respectively with no vertex of out-degree or in-degree four.
	Since the average out-degree of a reflexive tournament on $n>5$ vertices is $\frac {n-1} {2} +1> 3$, every reflexive tournament on at least six vertices has a vertex with out-degree at least four.
	Therefore if $T$ has at least four vertices, $T \neq T_4$ and $T \neq T_5$, then $T$ has  a vertex with either in-degree or out-degree at least four. 
	By repeated application of Lemma~\ref{indios} and Lemma~\ref{indios2} an instance of ios-injective homomorphism to $T$ polynomially transforms to instance of either ios-injective homomorphism to $T_4$, ios-injective homomorphism to $T_5$ or ios-injective homomorphism to a target on $3$ vertices. 
\end{proof}

\section{Iot-injective homomorphisms}\label{sec:iot}


In this section we prove a dichotomy theorem for iot-injective $T$-colouring, where $T$ is a reflexive tournament. 
We employ similar methods as in Section \ref{sec:ios}.
We first show that both iot-injective $T_4$-colouring and iot-injective $T_5$-colouring are NP-complete.
We  then provide a reduction to instances of iot-injective $T$-colouring to either  iot-injective $T_4$-colouring, iot-injective $T_5$-colouring, or a case covered by Theorem \ref{rus1}.
Combining these results yields the desired dichotomy theorem.

We begin with a study of iot-injective $T_4$-colouring.
To show iot-injective $T_4$-colouring is NP-complete we provide a transformation from 3-edge-colouring.
We construct an oriented graph $F$ from a  graph $G$ so that $G$ has a 3-edge-colouring if and only if $F$ admits an iot-injective homomorphism to $T_4$.
The key ingredients in this construction are the pair of oriented graphs $F_x$ and $F_{e}$, shown in Figure \ref{IOTT4VG}.

\begin{figure}[!ht]
\begin{center}
		\includegraphics{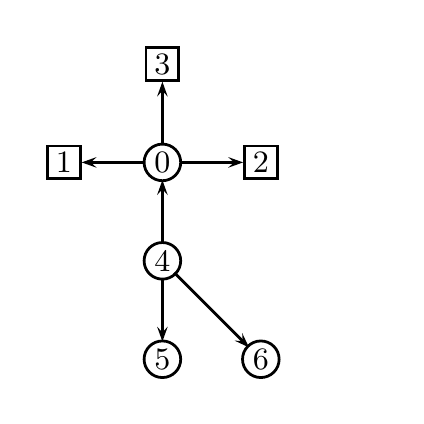}
	\includegraphics{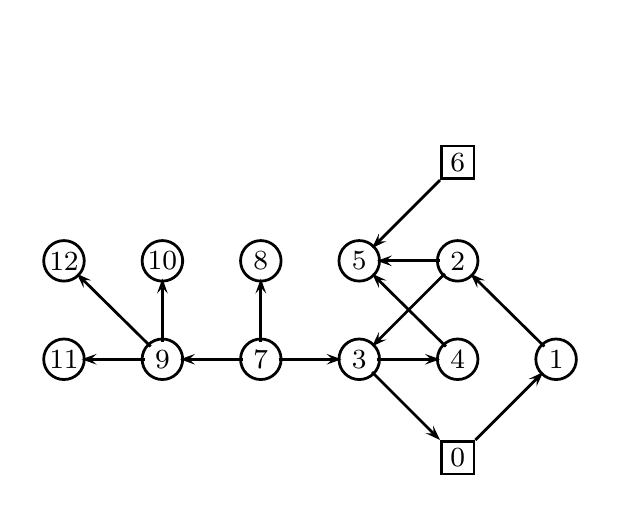}
	\caption{$F_x$ and $F_e$, respectively.}
	\label{IOTT4VG}
\end{center}
\end{figure}


\begin{lemma} \label{lem:vxgadgetIOT}
	In any iot-injective $T_4$-colouring of $F_x$, vertex $0$ is coloured $b$ and vertex $4$ is coloured $a$.
\end{lemma}

\begin{proof}
	Consider some iot-injective $T_4$-colouring of $F$.
	Vertex 0 of $F_x$ has out-degree $3$. 
	Since each vertex of $T_4$ has at most three out-neighbours (including itself), vertex $0$ must have the same colour as one of its out-neighbours.
	To satisfy the injectivity constraint, if a colour appears on an out-neighbour of vertex $0$, that colour cannot appear on an in-neighbour of vertex $0$.
	Therefore vertex $4$ does not have the same colour as vertex $0$.
	Both vertices $4$ and $0$ have out-degree $3$, and there is an arc from $4$ to $0$.
	Vertex $a$ in $T_4$ is the only vertex to have out-degree $3$ and have a strict out-neighbour with out-degree $3$.
	Therefore vertex $4$ is coloured $a$ and vertex $0$ is coloured $b$.
\end{proof}

\begin{lemma}  \label{lem:edgegadgetIOT1}
	In any iot-injective $T_4$-colouring of $F_{e}$, vertex $7$ is coloured $a$ and vertex $9$ is coloured $b$.
\end{lemma}

This proof of this lemma follows similarly to the proof of Lemma \ref{lem:vxgadgetIOT}. As such, it is omitted.

\begin{lemma}\label{lem:edgegadgetIOT2}
	Let $F_e^\prime$ be the oriented graph formed from a copy of $F_{e}$ and two copies of $F_x$ by identifying vertex $0$ in the copy of $F_e$ with any square vertex in one copy of $F_x$ and identifying vertex $6$ in the copy of $F_e$ with any square vertex in the second copy of $F_x$. In any iot-injective $T_4$-colouring of $F_e^\prime$, the vertices $0$ and $6$ in the subgraph induced by $F_e$ have the same colour, and are coloured with one of $b$, $c$ or $d$.
\end{lemma}

\begin{proof}
	Let $F_e^\prime$ be constructed as described. 
	Consider an iot-injective $T_4$-colouring of $F^\prime_e$. 
	We examine the colours of the vertices in the subgraph induced by the copy of  $F_e$. 
	By Lemma \ref{lem:vxgadgetIOT} and the construction of $F_e^\prime$, vertices $0$ and $6$ in $F_e$ have an in-neighbour coloured $b$ -- vertex $0$ in a copy of $F_x$.
	Since  $b$ is not an in-neighbour of $a$ in $T_4$, vertex $0$ in the copy of $F_e$ must receive one of the colours $b$, $c$, or $d$.
	We proceed in cases based on the possible colour of vertex $0$ in copy of $F_e$.
	
	\emph{Case I: Vertex $0$ is coloured $b$.}
	Since vertex 0 already has a neighbour coloured $b$  (vertex $0$ in a copy of $F_x$),  vertex $3$, an in-neighbour of vertex $0$ in $F_e$, cannot be coloured $b$.
	Since $b \in V(T_4)$ has only $a$ and $b$ as in-neighbours, we have that vertex $3$ is coloured $a$.
	By Lemma \ref{lem:edgegadgetIOT1} vertex $3$ has a neighbour coloured $a$ -- namely vertex $7$.
	By the injectivity constraint, this colour cannot appear on any other neighbour of vertex $3$. 
	As such, vertices $2$ and $4$ are coloured $d$ and $c$ respectively.
	The only common out-neighbour of $d$ and $c$ in $T_4$ is $d$.
	Therefore vertex $5$ has colour $d$.
	In $T_4$ vertex $d$ has three in-neighbours -- $b$, $c$ and $d$.
	Since $c$ and $d$ both appear on an in-neighbour of vertex $5$, we have that vertex $6$ is coloured with $b$.
	
	\emph{Case II: Vertex $0$ is coloured $c$.}
	Vertex $c$ in $T_4$ has three in-neighbours: $a$, $b$ and $c$.
	Since vertex $0$ has an in-neighbour coloured $b$, namely vertex $0$ in a copy of $F_x$, vertex $3$ in $F_e$ must have either colour $a$ or colour $c$.
	
	Assume vertex $3$ is coloured $c$. 
	In this case, the injectivity constraint implies that vertex $1$ is not coloured $c$.
	Since $c$ and $d$ are the only out-neighbours of $c$ in $T_4$, vertex $1$ must be coloured $d$.
	Vertex $2$, an in-neighbour of vertex $3$, is coloured with one of $a,b$ or $c$, the in-neighbours of $c$ in $T_4$.
	By Lemma \ref{lem:edgegadgetIOT1} vertex $3$ has a neighbour coloured $a$ --  vertex $7$.
	By assumption vertex $0$ has colour $c$.
	Therefore by injectivity, vertex $2$ has colour $b$.
	This is a contradiction, as the arc between vertex $1$ and vertex $2$ does not have the same direction as the arc between vertex $c$ and vertex $b$ in $T_4$.
	Therefore vertex $3$ is coloured $a$.
	
	In $T_4$, the in-neighbours of $a$ are $a$ and $d$, and the out-neighbours of $a$ are $a,b$ and $c$.
	Since vertex $7$ is coloured $a$, no other neighbour of vertex $3$ can be coloured $a$.
	Therefore vertex $2$, an in-neighbour of vertex $3$, must have colour $d$.
	Since vertex $0$ is coloured $c$, vertex $4$, an out-neighbour of vertex $3$ must have colour $b$.
	Vertex $5$, a common in-neighbour of  vertices $2$ and $4$, must be coloured with a common in-neighbour of $d$ and $b$ in $T_4$.
	The only such vertex in $T_4$ is $d$.
	Therefore vertex $5$ has colour $d$.
	
	Vertex $d$ in $T_4$ has three in-neighbours: $b$, $c$ and $d$.
	Since vertex $2$ is coloured $d$ and vertex $4$ is coloured $b$, we have that vertex $6$ is coloured $c$, as required.
	
	\emph{Case III: Vertex $0$ is coloured $d$.}
	Recall by Lemma \ref{lem:edgegadgetIOT1} that vertex $3$ has a neighbour coloured $a$ --  vertex $7$.
	Since vertex $0$ is coloured $d$, vertex $3$ is coloured with a vertex that is an out-neighbour of $a$ and an in-neighbour of $d$ in $T_4$.
	The only such vertices are $b$ and $c$.
	However, vertex $0$ has a neighbour coloured $b$ (vertex $0$ in a copy of $F_x$).
	Therefore vertex $3$ has colour $c$.
	Vertex $4$ must have a colour that is an out-neighbour of $c$ in $T_4$.
	The only such colours are $c$ and $d$.
	Since vertex $0$, an out-neighbour of vertex $3$, is coloured $d$, we have that vertex $4$ has colour $c$.
	Vertex $2$ must have a colour that is an in-neighbour of $c$ in $T_4$.
	The only such colours $a,b$ and $c$.
	Vertex $7$, an in-neighbour of vertex $3$, has colour $a$.
	Vertex $4$, a neighbour of vertex $3$, has colour $c$.
	Therefore by injectivity vertex $2$ has colour $b$.
	Vertex $5$ must be coloured with a common out-neighbour of $b$ and $c$ in $T_4$.
	The only such colours are $c$ and $d$.
	Since vertex $3$, an out-neighbour of vertex $2$, has colour $c$, we have by injectivity that vertex $5$ has colour $d$.
	The in-neighbours of vertex $5$ must be coloured with the in-neighbours of $d$ in $T_4$.
	Vertex $d$ has three in-neighbours in $T_4$ -- $b,c$ and $d$.
	Since vertex $2$ has colour $b$ and vertex $4$ has colour $c$, we have by injectivity that vertex $6$ has colour $d$.
\end{proof}

\begin{theo} 
	The problem of  iot-injective $T_4$-colouring is NP-complete.
\end{theo}

\begin{proof}
	The transformation is from 3-edge-colouring subcubic graphs \cite{Holyer}.	
	
	Let $G$ be a graph with maximum degree at most $3$ and let $\tilde{G}$ be an arbitrary orientation of $G$.
	We create an oriented graph $F$ from $\tilde{G}$ as follows.
	For every $v \in V(G)$ we add $F_v$, a copy of the oriented graph $F_x$ given in Figure \ref{IOTT4VG}, to $F$.
	For every arc $uv \in E(\tilde{G})$ we add $F_{uv}$, a copy of the oriented graph $F_e$ given in  Figure \ref{IOTT4VG}, to $F$.
	To complete the construction of $F$, for each arc $uv \in E(\tilde{G})$ we identify vertex $0$ in $F_{uv}$ with one of the three square vertices (i.e., vertices $1$, $2$, or $3$) in $F_u$ and identify vertex $6$ in $F_{uv}$ with one of the three square vertices in $F_v$.
	We identify these vertices in such a way that each square vertex in a copy of $F_x$ is identified with at most one square vertex from a copy of $F_e$.
	We note that this is always possible as vertices in $G$ have degree at most three.
	
	We claim that $G$ has a 3-edge-colouring if and only if $F$ has an iot-injective $T_4$-colouring.	
	
	Suppose an iot-injective $T_4$-colouring of $F$ is given.
	This iot-injective $T_4$-colouring induces a 3-edge-colouring of $G$: the colour of an edge in $uv \in E(G)$ is given by colour of vertices $0$ and $6$ in corresponding copy of $F_{uv}$ contained in $F$. 
	By Lemma \ref{lem:edgegadgetIOT2} this colour is well-defined, and is one of $b$, $c$, or $d$.
	Recall for each copy of $F_x$, the vertices 1,2 and 3 are respectively each identified with either vertex 0 or vertex 6 in some copy of $F_{e}$.
	By Lemma \ref{lem:vxgadgetIOT}, vertices $1$, $2$ and $3$ in a copy of $F_x$ cannot be coloured  $a$.
	By injectivity, vertices $1$, $2$ and $3$ in a copy of $F_x$ all are assigned different colours.
	Therefore each of the edges incident with any vertex receives different colours and no more than $3$ colours, namely $b$, $c$, and $d$, are used on the edges of $G$.
	Therefore $G$ has a 3-edge-colouring.
	

	Suppose a 3-edge-colouring of $G$, $f: E(G) \to \{b,c,d\}$ is given.
	For each $uv \in E(G)$ we colour $F_{uv}$ using one of the iot-injective $T_4$-colourings given in Figure \ref{fig:edgegadgT4iot}. 
	We choose the colouring of $F_{uv}$ so that vertices $0$ and $6$ are assigned the colour $f(uv)$. 
	To complete the proof, we show that such colouring can be extended to all copies of $F_x$ contained in $F$.
	
	Recall for each copy of $F_x$, the vertices 1, 2 and 3 are respectively each identified with either vertex 0 or vertex 6 in some copy of $F_{e}$.
	Since $f$ is a 3-edge-colouring of $G$, for each $x \in V(G)$, each of the vertices $1$, $2$ and $3$ in $F_x$ are coloured with distinct colours from the set $\{b,c,d\}$ when we colour each copy of $F_{e}$ using Figure~\ref{fig:edgegadgT4iot}.
	
	By symmetry of $F_x$, we may assume without loss of generality that vertices $1,2$ and $3$ are respectively coloured $b$, $c$ and $d$ in each copy of $F_x$.
	The iot-injective $T_4$-colouring given in Figure \ref{vertiott4} extends a pre-colouring of the vertices 1, 2 and 3 with colours $b$, $c$, and $d$, respectively, to an iot-injective $T_4$-colouring of $F_x$. 
	Therefore $F$ has an iot-injective $T_4$-colouring.
	
	\begin{figure}
\begin{center}
			\includegraphics{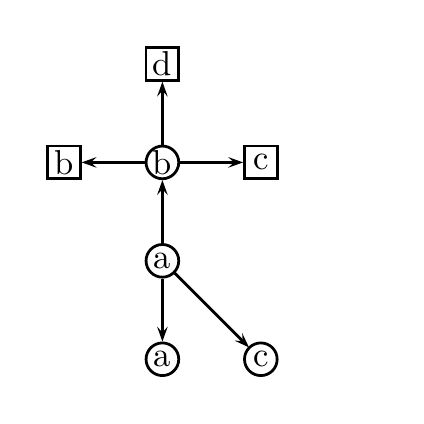}
		\caption{Colouring  $F_x$.}
		\label{vertiott4}
\end{center}
	\end{figure}
	
	\begin{figure}
	\begin{center}
			\includegraphics[scale=0.8]{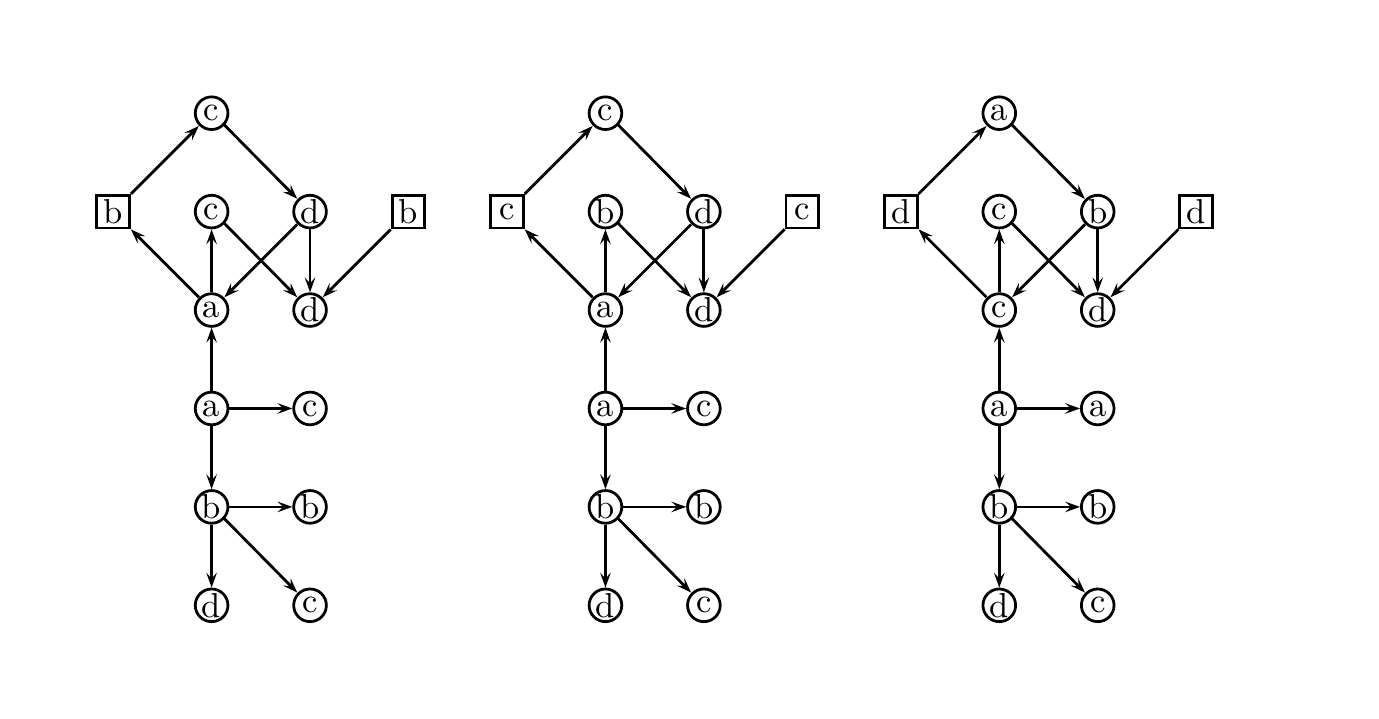}
		\caption{\label{fig:edgegadgT4iot}Colouring the $F_e$ for associated edge colours of $b$, $c$, and $d$.}
	\end{center}
	\end{figure}
	
	Since the construction of $F$ can be carried out in polynomial time, iot-injective $T_4$-colouring is NP-complete.
\end{proof}

We provide a similar argument for iot-injective $T_5$-colouring. The transformation is from iot-injective $C_3$-colouring \cite{Rus1}. We construct an oriented graph $D$ from a  graph $G$ so that $G$ admits an iot-injective homomorphism to $C_3$ if and only if $D$ admits an iot-injective homomorphism to $T_5$. The key ingredient in the construction is the oriented graph, $D_v$, given in Figure \ref{Fv}.

\begin{figure}[!ht]
\begin{center}
		\includegraphics{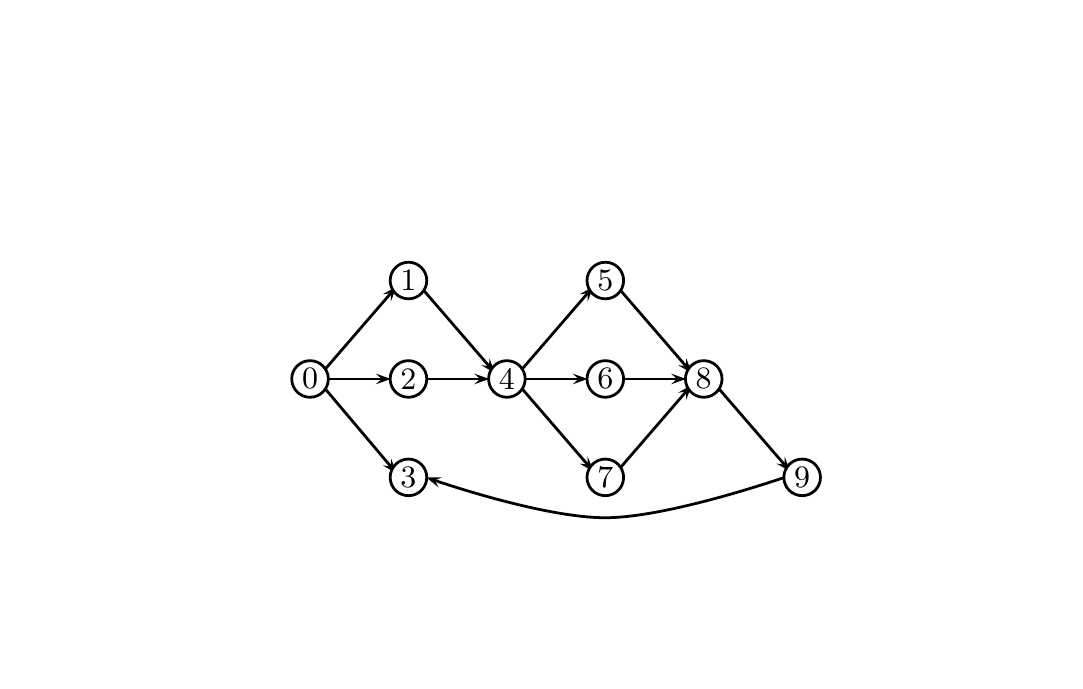}
	\caption{$D_v$.}
	\label{Fv}
\end{center}
\end{figure}

For each $n>0$ let $D_n$  be the oriented graph constructed from $n$ disjoint copies of $D_v$, say $D_{v_0}, D_{v_1},$ $ \dots,D_{v_{n-1}}$, 
by letting vertex $8$ of $D_{v_i}$ be an in-neighbour of vertex $0$ in $D_{v_{i+1}\pmod n}$ for all $0 \leq i \leq n-1$.

\begin{lemma} \label{lem:loopLemmaiot}
	For any positive integer $n$, up to automorphism, in an oriented iot-injective $T_5$-colouring of $D_n$ each of the vertices labeled $0$ receive the colour $d$, each of the vertices labeled $4$ receive the colour $a$, and each of the vertices labeled $8$ receive the colour $c$.
\end{lemma}

\begin{proof}
	Since $T_5$ is vertex-transitive, assume without loss of generality that vertex $0$ in 	$D_{v_0}$ receives colour $d$ in some iot-injective $T_5$-colouring of $D_n$.
	Observe that vertex $4$ has three out-neighbours. 
	Since each vertex of $T_5$ has at most three out-neighbours (including itself), vertex $4$ must have the same colour as one of its out-neighbours.
	To satisfy the injectivity constraint,  no in-neighbour of vertex $4$ has the same colour as vertex $4$.
	Further, vertex $4$ has two in-neighbours, vertices $1$ and $2$, that are out-neighbours of a vertex coloured $d$.
	Only vertices $a$ and $b$ in $T_5$  have two in-neighbours that are out-neighbours of $d$.
	Therefore vertex $4$ has colour $a$ or $b$.
	
	If vertex $4$ has colour $b$, then vertices $1$ and $2$ are coloured with the vertices of $T_5$ that are out-neighbours of $d$ and in-neighbours of $b$. 
	The only such vertices in $T_5$ that satisfy these criteria are $a$ and $e$. 
	Therefore vertices $1$ and $2$ are coloured with $a$ and $e$, in some order.
	Vertex $d$ has three out-neighbours in $T_5$ -- $a,d$ and $e$.
	Since vertices $1$ and $2$ are coloured with $a$ and $e$, in some order, the third out-neighbour of vertex $0$, vertex $3$, is coloured with $d$.
	Vertex $b$ in $T_5$ has three out-neighbours -- $b,c$ and $d$.
	Therefore the out-neighbours of vertex $4$, vertices $5,6$ and $7$, are coloured, in some order, with these colours.
	Vertices $b,c$ and $d$ in $T_5$ have only $d$ as a common out-neighbour.
	Therefore the common out-neighbour of vertices $5,6$ and $7$, vertex $8$, is coloured $d$.
	This is a contradiction, as now vertex $9$ has two vertices coloured $d$ in its neighbourhood.
	Therefore vertex $4$ has colour $a$.

	Vertex $a$ in $T_5$ has three out-neighbours -- $a,b$ and $c$. 
	Thus the out-neighbours of vertex $4$ are coloured with $a$, $b$ and $c$, in some order.
	The only common out-neighbour of 	$a,b$ and $c$ in $T_5$ is $c$.
	Therefore vertex $8$ has colour $c$.
	This implies that vertex $9$ in $D_{v_0}$ and $0$ in $D_{v_1}$ have colours from the set  $\{c,d,e\}$, the out-neighbours of $c$ in $T_5$.
	Since vertex 8  has a neighbour coloured $c$, neither vertex $0$ in $D_{v_1}$ nor $9$ (in $D_{v_0}$)  can have this colour.
	Further, since vertex 3 has a neighbour coloured $d$,  vertex 9 has cannot be coloured $d$.
	Thus vertex $9$ in $D_{v_0}$  has colour $e$ and vertex $0$ in $D_{v_1}$ has colour $d$.
	
	Repeating this argument implies that every vertex labeled $0$ has colour $d$.
\end{proof}

\begin{theo}
	The problem of iot-injective $T_5$-colouring is NP-complete.
\end{theo}

\begin{proof}
	
	The transformation is from iot-injective $C_3$-colouring (See Theorem~\ref{rus1}).

	Let $G$ be an oriented graph with vertex set $\{v_0,v_1, \dots, v_{|V(G)|-1}\}$.
	Let $\nu_G = |V(G)|$.
	We construct $D$ from $G$ by first adding a copy of $D_{\nu_G}$ to $G$  and then, for each $1 \leq i \leq \nu_G$, adding an arc from vertex $5$ in $D_{v_i}$ to $v_i$.	
	
	We show that $D$ has an  iot-injective $T_5$-colouring if and only if $G$ has an iot-injective $C_3$-colouring.
	
	Consider $\phi$, an iot-injective $T_5$-colouring of $D$. 
	Since $T_5$ is vertex-transitive we may assume that vertex $0$ in $D_{v_0}$ has colour $d$.
	By Lemma \ref{lem:loopLemmaiot}, for all $0 \leq i \leq \nu_G-1$, the vertex in $D_{v_i}$ labeled $0$ has colour $d$, the vertex labeled $4$ has colour $a$ and the vertex labeled $8$ has colour $c$.
	By the injectivity requirement, the neighbours of the vertex labeled $5$ in each copy of   $D_{v}$ have distinct colours.
	Since the vertices $4$ and $8$ have colours $a$ and $c$, respectively, only colours $b,d$ or $e$ can appear at $v_i$, for all $0 \leq i \leq \nu_G-1$.
	Since $b,d$, and $e$ induce a copy of $C_3$ in $T_5$, we conclude that the restriction of $\phi$ to the vertices of $G$ is indeed an  iot-injective $C_3$-colouring.
	
	Let $\beta$ be an iot-injective $C_3$-colouring of $G$ using colours $b,d$ and $e$. We extend such a colouring to be an ios-injective $T_5$-colouring of $D$ by assigning the vertices of each $D_{v_i}$ colours based upon $\beta(v_i)$ as shown in Figure~\ref{fig:iotT5}. This colouring satisfies the injectivity requirement, as each vertex $v_i$ has only neighbours coloured $b,d$ and $e$ in $G$, and its additional neighbour in $D_{v_i}$, vertex $5$, has colour $a$ or $c$.
	
	\begin{figure}[!ht]
	\begin{center}
			\[\includegraphics[width = .4\linewidth]{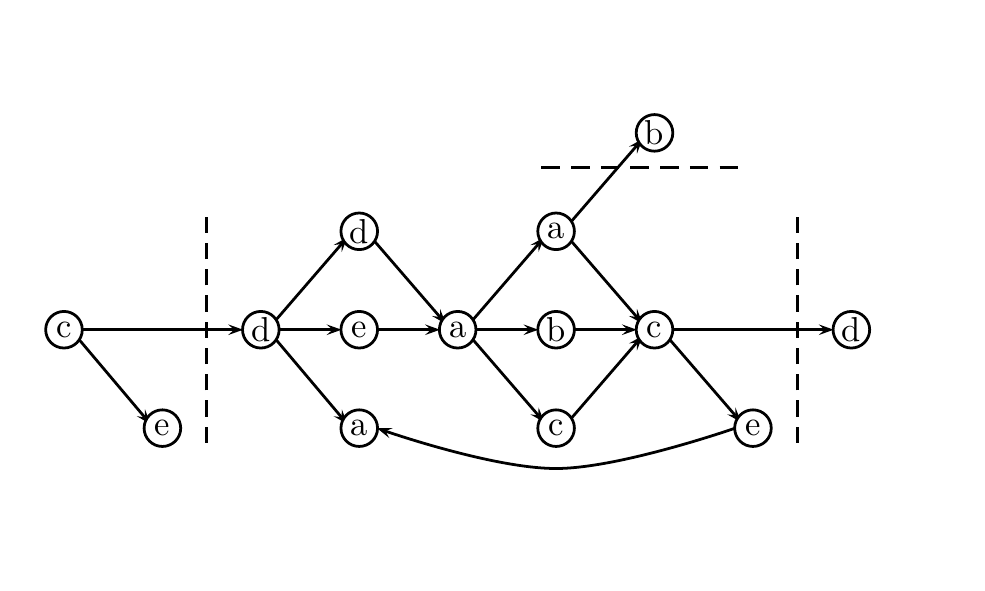}
		\includegraphics[width = .4\linewidth]{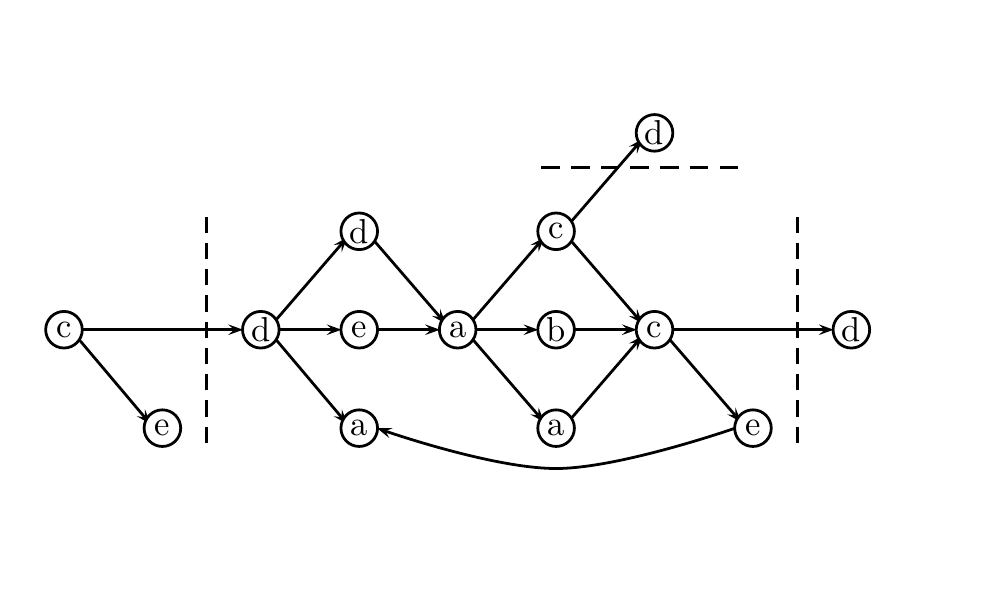}\]
		\[	\includegraphics[width = .4\linewidth]{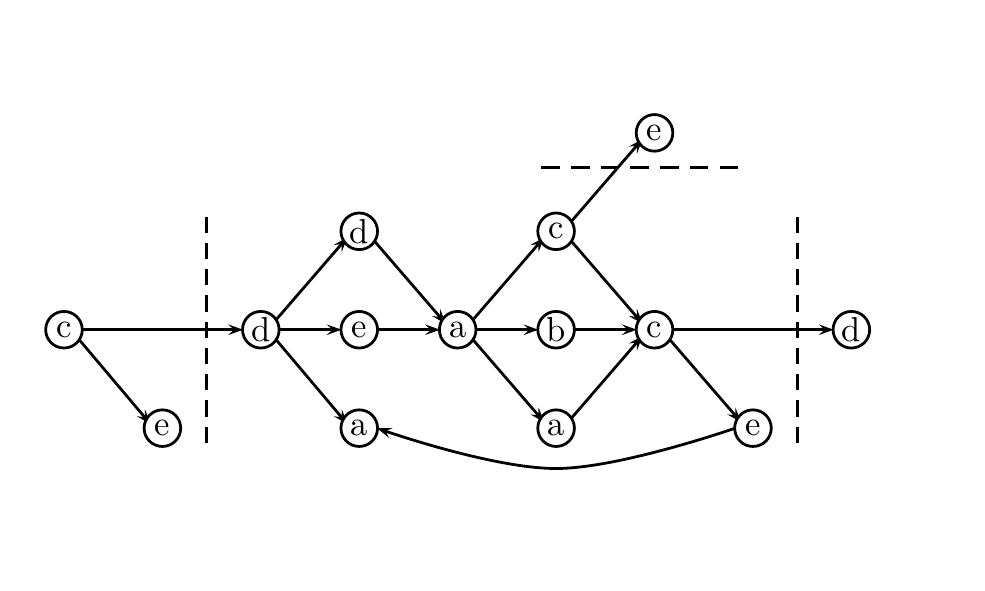}\]
		\caption{\label{fig:iotT5} Colouring the vertices of $D_{v_i}$ using the colour of $v_i$.}
	\end{center}
	\end{figure}

	Therefore $D$ has an iot-injective $T_5$-colouring if and only if $G$ has an iot-injective $C_3$-colouring. 
	As $D$ can be constructed in polynomial time,  iot-injective $T_5$-colouring is NP-complete.
\end{proof}

We now present a reduction to instances of iot-injective $T$-colouring  for when $T$ has a vertex $v$ of out-degree at least four. This reduction allows us to polynomially transform an instance of iot-injective $T$-colouring to an instance of iot-injective $T^\prime$-colouring, where $T^\prime$ is $T_4$, $T_5$, $C_3$ or $T\!T_3$.

\begin{lemma}
	\label{indiot}
	If $T$ is a reflexive tournament on $n$ vertices with a vertex $v$ of out-degree at least four, then iot-injective homomorphism to $T'$ polynomially transforms to iot-injective homomorphism to $T$, where $T'$ is the tournament induced by the strict out-neighbourhood of $v$. 
\end{lemma}

\begin{proof}
	Let $T$ be a reflexive tournament on $n$ vertices with a fixed vertex $v$ of out-degree four or more. 
	Let $T^\star$ be the graph obtained by removing from $T$ all the arcs with their tail at $v$.
	
	Let $G$ be an oriented graph with vertex set $\{w_0,w_1, \dots,w_{|V(G)|-1}\}$.
	Let $\nu_G = |V(G)|$.
	We construct $C$ from $G$ by adding to $G$
	\begin{itemize}
		\item $\nu_G$  disjoint irreflexive copies of $T: T_0,T_1, \dots,T_{\nu_G-1}$;
		\item $\nu_G$  disjoint irreflexive copies of $T^\star: T^\star_0,T^\star_1, \dots,T^\star_{\nu_G-1}$;
		\item and for all $u \in V(T)$ where $u \neq v$, an arc from the vertex corresponding to $u$ in $T^\star_{i-1}$ to the vertex corresponding to $u$ in $T_i$, for all $0 \leq i \leq {\nu_G-1}$
	\end{itemize}
	
	Let $v_i$ and $v_i^\star$ be the vertices corresponding to $v$ in $T_i$ and $T_i^\star$, respectively. 
	We complete the construction of $C$  by adding an arc from $v_i$ to $v_i^\star$ for all $0 \leq i \leq {\nu_G-1}$.
	See Figure \ref{iotlem}
	
	\begin{figure}[!ht]
\begin{center}
			\includegraphics{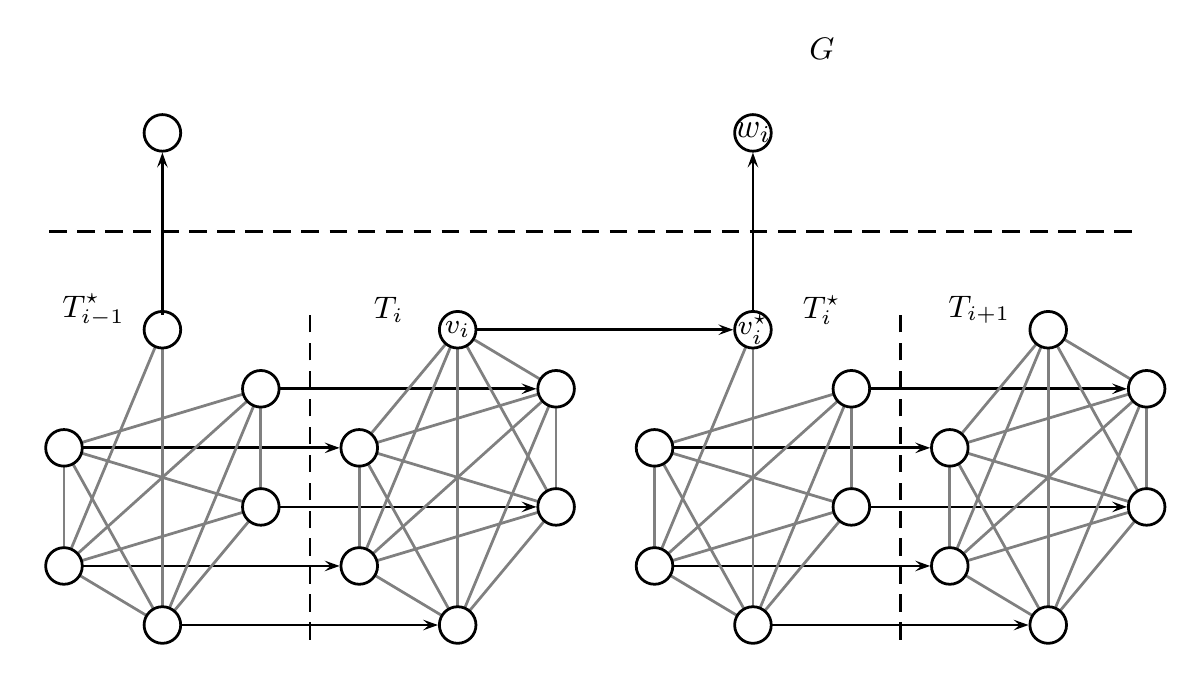}
		\caption{The construction of $C$ in Lemma \ref{indiot}.}
		\label{iotlem}
\end{center}
	\end{figure}
	
	\emph{Claim 1: In an iot-injective $T$-colouring of $C$, no two vertices of $T_i$ have the same colour}.
	If two vertices of $T_{i}$ are assigned the same colour, then a common neighbour of such vertices in $T_{i}$ has a pair of neighbours with the same colour. This is a violation of the injectivity requirement.
	Therefore no two vertices of $T_{i}$ are assigned the same colour.
	
	\emph{Claim 2: In an iot-injective $T$-colouring of $C$, $v_i$ and $v_i^\star$ have the same colour}.
	Since $v_i$ has $n$ neighbours, in any iot-injectve $T$-colouring of $C$, $v_i$ is assigned the same colour as one of its neighbours. By the previous claim, the neighbour of $v_i$ that has the same colour as $v_i$ must be $v_i^\star$.
	
	\emph{Claim 3: In an iot-injective $T$-colouring of $C$, $v_i$ and $v_{i+1}$ have the same colour}.
	We show that $v_{i+1}$ has the same colour as $v_i^\star$.
	If $v$ has out-degree $n$ in $T$, then by construction $v_{i}$ and $v_{i+1}$ each have out-degree $n$ in $C$.
	Since no two out-neighbours of any vertex can receive the same colour, and since there can be at most one vertex of out-degree $n$ in $T$, it must be that $v_i$ and $v_{i+1}$ have colour $v$.
	Suppose then that $v$ has at least one in-neighbour distinct from itself, say $y$, in $T$. 
	Let $y_i^\star$ be the vertex corresponding to $y$ in $T_i^\star$. 
	Let $u_{i+1}$ be a vertex in $T_{i+1} \setminus \{v_{i+1}\}$, and let $u_i^\star$ be the vertex of  $T_i^\star$ which has $u_{i+1}$ as an out-neighbour. 
	By the first claim, no two vertices in $T_{i+1}$ share a colour, and since $u_{i+1}$ has $n-1$ neighbours in $T_{i+1}$, it must be that $u_{i+1}$ and $u_i^\star$ share a colour. 
	This implies that no two vertices of $T_i^\star \setminus \{v_i^\star\}$ have the same colour, and the colours used to colour $T_i^\star \setminus \{v_i^\star\}$ are the same colours as those used to colour $T_{i+1} \setminus \{v_{i+1}\}$. 
	The vertex $y_i^\star$ has $v_i^\star$ as an out-neighbour, and each colour except the colour of $v_{i+1}$ is used to colour a vertex distinct from $v_i^\star$ which is a neighbour of $y_i^\star$. 
	Therefore, $v_i^\star$ must have the same colour as $v_{i+1}$. The result now follows from the previous claim.
	\vspace{\baselineskip}

	Let $T^\prime$ be the reflexive sub-tournament of $T$ induced by the strict out-neighbourhood of $v$.
	We show $G$ has an iot-injective $T'$-colouring if and only if $C$ has an iot-injective $T$-colouring.
	
	Consider an iot-injective $T$-colouring of $C$, $\phi$.
	By the claims above, $\phi(v_i^\star) = \phi(v_i) = \phi(v_j) = \phi(v_j^\star)$ for all $1 \leq i,j \leq \nu_G -1$.
	Since each vertex in $T_i$ is assigned a distinct colour from $T$ and $T \cong T_i$, if $\phi(v_i) \neq v$, then there is an automorphism of $T$ that maps $v$ to $\phi(v_1)$.
	As such we may assume, without loss of generality that $\phi(v_i) = v$ for all $1 \leq i \leq \nu_G -1$.
	Since $w_i$ is an out-neighbour of $v_i^\star$ for each  $1 \leq i, \leq \nu_G -1$ we have that $\phi(w_i)$ is contained in the out-neighbourhood of $v$ for all $1 \leq i \leq \nu_G -1$.
	That is, $\phi(w_i) \in V(T^\prime)$ for all $1 \leq i \leq \nu_G -1$.
	Therefore the restriction of $\phi$ to $G$ is an iot-injective homomorphism to $T^\prime$.
	
	Consider now an iot-injective $T^\prime$-colouring of $G$, $\beta$.
	We extend $\beta$ to be an iot-injective $T$-colouring of $C$ as follows.
	For each $z \in V(T)$ let $z_i$ and $z_i^\star$ be the corresponding vertices in $T_i$ and $T_i^\star$, respectively.
	We extend $\beta$ so that $\beta(z_i) = \beta(z_i^\star) = z$.
	It is easily verified that $\beta$ is an iot-injective $T$-colouring of $C$.
\end{proof}

The construction of $C$ can be modified to give the corresponding result for reflexive tournaments $T$ with a vertex of in-degree at least four.

\begin{lemma}
	\label{indiot2}
	If $T$ is a reflexive tournament on $n$ vertices with a vertex $v$ of in-degree at least four, then iot-injective homomorphism to $T'$ polynomially transforms to iot-injective homomorphism to $T$, where $T'$ is the tournament induced by the strict in-neighbourhood of $v$. 
\end{lemma}

Similar to the case of ios-injective colouring, our results compile to give a dichotomy theorem.
\begin{theo}
	\label{iotdichot}
	Let $T$ be a reflexive tournament. If $T$ has at least $3$ vertices, then the problem of deciding whether a given oriented graph $G$ has an iot-injective homomorphism to $T$ is NP-complete. If $T$ has $1$ or $2$ vertices, then the problem is solvable in polynomial time.
\end{theo}

\section{A note on irreflexive-injective homomorphisms}

No dichotomy theorem has emerged yet for iot-injective homomorphism, and hence ios-injective homomorphism, to irreflexive tournaments. The results of \cite{Rus1, Russell} tell us that the problem is not only solvable in polynomial time for the irreflexive tournaments on two vertices or less but also for the irreflexive tournaments on three vertices. Preliminary work suggests that the problem is solvable in polynomial time on two of the irreflexive tournaments on four vertices but NP-complete on the remaining two, and on many irreflexive tournaments on more vertices (including at least ten of the twelve irreflexive tournaments on five vertices). The problem has not been proven solvable in polynomial time on any irreflexive tournament on five vertices or more.

\nocite{*}
\bibliographystyle{abbrvnat}
\bibliography{sample-dmtcs}
\label{sec:biblio}

\end{document}